\documentclass[12pt]{article}

\usepackage{latexsym,amsfonts,amsmath,
amssymb,epsfig,tabularx,amsthm,dsfont,enumerate,graphicx, setspace, multirow, makecell}
\usepackage{algorithm} 
\usepackage{algorithmic} 
\usepackage{subfigure}
\usepackage[square, comma, sort&compress, numbers]{natbib} 
\usepackage{color}

\usepackage{epstopdf}

\makeatletter
\newcommand{\LeftEqNo}{\let\veqno\@@leqno}
\makeatother

\theoremstyle{definition}
\newtheorem{thm}{Theorem}
\newtheorem{rem}[thm]{Remark}

\newtheorem{defi}[thm]{Definition}

\newcommand\eu{{\rm e}}

\newcommand{\widebar}[1]{\mbox{\kern1pt\hbox{
\vbox{\hrule height 0.5pt \kern0.25ex
        \hbox{\kern-0.05em \ensuremath{#1 }\kern-0.05em}}}}\kern-0.1pt}

\definecolor{darkblue}{RGB}{40,0,200}

\newlength{\fixboxwidth}
\renewcommand{\rho}{{\varrho }}

\title{Probabilistic Matrix Factorization with Personalized Differential Privacy} 

\author{Shun Zhang,\ \
Laixiang Liu,\ \ Zhili Chen\footnote{Corresponding author.
\newline\indent\, Email addresses: shzhang27@163.com (S. Zhang), ahu\_laixiangliu@163.com (L. Liu),
zlchen@ahu.edu.cn (Z. Chen), zhongh@ahu.edu.cn (H. Zhong).} ,\ \
Hong Zhong \ \
\\ {\footnotesize \rm School of Computer Science and Technology, Anhui
University, Hefei 230601, China}
 }
\begin{document}

\maketitle

\begin{abstract}
Probabilistic matrix factorization (PMF) plays a crucial role
in recommendation systems. It requires a large amount of user data
(such as user shopping records and movie ratings) to predict personal preferences, and
thereby provides users high-quality recommendation services,
which expose the risk of leakage of user privacy.
Differential privacy, as a provable privacy protection framework,
has been applied widely to recommendation systems. It is common that different individuals have different levels of privacy requirements on items. However,  traditional differential privacy
can only provide a uniform level of privacy protection for all users.

In this paper, we mainly propose a probabilistic matrix factorization recommendation scheme
with personalized differential privacy (PDP-PMF). It aims to meet users' privacy requirements
specified at the item-level instead of giving the same level of privacy guarantees for all. We
then develop a modified sampling mechanism (with bounded differential privacy) for achieving PDP.
We also perform a theoretical analysis of
the PDP-PMF scheme and demonstrate the privacy of the PDP-PMF scheme.
In addition, we implement the probabilistic matrix factorization schemes both
with traditional and with personalized differential privacy (DP-PMF, PDP-PMF) and compare them through a series of experiments.
The results show that the PDP-PMF scheme performs well on protecting the privacy of each user
and its recommendation quality is much better than the DP-PMF scheme.
\end{abstract}

\section{Introduction}

In the last decade, with the rapid development of web-based applications, information has grown explosively.
It is impossible for people to extract relevant data by exploring all web content. In order to provide
 personalized recommendation services, recommendation systems \cite{KR12} are used widely and promoted rapidly. For building
recommendation systems, probabilistic matrix factorization (PMF) is a prevailing method \cite{SM08} that performs well on large and sparse datasets.
This model regards the user
preference matrix as a product of two lower-rank user and item matrices and adopts a probabilistic linear setting with Gaussian observation
noise.
Training such a model amounts to finding the best
low-rank approximation to the observed target matrix under a given loss function.
During the whole recommendation process, a large amount of
user information has to be utilized to train the recommendation algorithm.
Such information may include users' privacy (such as shopping records and
rating records), and there is a probability of information leakage. Therefore, such a recommendation system
is a double-edged sword. It may promote the rapid development of the internet economy
and there is a risk of leaking sensitive information.

Differential privacy (DP) \cite{Dwo08} is a provable and strict privacy protection framework. Usually, it keeps
private information confidential by perturbing the output of the algorithm with noise.
So far, the DP method has been widely applied to recommendation algorithms \cite{FM09,ABK16}.
It is worth noting that the objective perturbation method proposed in \cite{CM09}
is utilized in \cite{HXZ15} to achieve (traditional) differential privacy for matrix factorization recommendations. We find that this traditional DP method
is also fit for the PMF scheme. Our main idea, for the probabilistic matrix factorization scheme with traditional differential privacy (DP-PMF), is firstly to obtain the user profile matrix
through the PMF algorithm without privacy protection and keep it confidential. Then the user profile matrix
is input as a constant into the objective (perturbation) function (different from \cite{HXZ15}), and the resulting perturbed
item profile matrix satisfies $\epsilon$-differential privacy ($\epsilon$-DP). Finally, the item profile
matrix is released. This makes it impossible for an attacker to guess specific user's private rating data
based on the recommendation results. Thereby the user's private information is protected.

However, the traditional differential privacy is aimed at providing the same level of privacy protection for all users.
This ``one size fits all'' approach ignores the reality that:  users' expectations on privacy protection are not uniform for various items in general \cite{BGS05,JYC15}. Indeed, users would like to choose freely different levels of privacy protection for their ratings and other data.
Unfortunately, the traditional differential privacy method can only provide the highest level (the smallest $\epsilon$)
privacy protection for all users in this case, so as to meet the privacy requirements of all users \cite{JYC15}.
This generates excessive noise and thus impairs seriously
the accuracy of the recommendation results. Afterwards it is impossible to provide users with
high-quality recommendation services. To address this problem, personalized differential privacy (PDP) was introduced into recommendation systems  in \cite{LLW17,YZX17}
where a PDP protection
framework was proposed to protect user-based collaborative filtering recommendation algorithms.
Further, due to the sparsity of the user-item rating matrix, user-based collaborative filtering algorithms therein may fail to find similar users \cite{MYL08,WVR06}.
Thus in this case the accuracy of recommendation is probably poor. It has been noted that the PMF model can overcome the above problem and get a better accuracy \cite{SM08}.

In this paper, we introduce the concept of personalized differential privacy (PDP) to the task of PMF. The proposed
probabilistic matrix factorization recommendation scheme with personalized differential privacy (PDP-PMF)
achieves item-level privacy protection. It meets users' requirements for different levels of privacy protection on various items, while achieving a good recommendation quality. The main contributions are as follows:
\begin{enumerate}[(1)]
  \item
  We consider the scenario that every user may have potentially different privacy requirements for all items and the recommender is trustworthy.
The goal here is to ensure that the recommendation scheme designed satisfies
 personalized differentially private requirements $\epsilon_{ij}$.
We present a PDP-PMF scheme using a so-called bounded sample mechanism, which guarantees users' privacy requirements and achieves item-level privacy protection as desired.

  \item We conduct a theoretical analysis of PDP-PMF as well as DP-PMF. Based on the definition of
  personalized differential privacy that was first presented in \cite{JYC15}, we prove that
  the PDP-PMF scheme meets personalized differential privacy and guarantees the privacy of the whole model.
  \item We carry out multiple sets of experiments on three real datasets for comparisons of the PDP-PMF scheme with DP-PMF.
 The experimental results show that PDP-PMF protects the data privacy of each user well and the recommendation quality is much better than that of DP-PMF.
\end{enumerate}

The remainder of this paper is organized as follows. In Section 2, we briefly recall the related work. Section 3 introduces probabilistic matrix factorization technique, (personalized) differential privacy and our setting for the schemes to design. We provide the detailed design of our traditional differentially private recommendation scheme in Section 4. In Section 5, the crucial part of the
work will be done, we present the personalized differentially private recommendation scheme. Section 6 demonstrates and analyzes the performance of both traditional and personalized differentially private schemes with experiments on three public datasets. Finally, we conclude this paper in Section 7.


\section{Related work}
\textbf{Traditional Privacy-preserving Recommendations.}
In recommendation systems, the traditional ways to protect user privacy data usually include
cryptography and obfuscation based approaches. The cryptology based approach is suitable for the protection of recommendation systems
composed of mutually distrustful multiparty, and it normally results in significantly high computational costs \cite{Can02,AS11}, especially when
the amount of historical data is large. On the other hand, the obfuscation based method can improve the computation efficiency effectively,
but it can not resist the attack of background knowledge and may weaken the utility of data \cite{BEK07,PB07}.


\textbf{Differentially Private Recommendations.}
Differential privacy proposed by Dwork becomes a hot privacy framework recently, and many papers \cite{Dwo08,POV17,ZQZ17,XWG11} have been published on
differential privacy.
Some literatures aimed to utilize differential privacy to protect recommendations, and demonstrated that the
differentially private recommendation can solve perfectly some traditional privacy-preserving problems.
In particular, McSherry et al. \cite{FM09} first applied
differential privacy to collaborative filtering recommendation algorithms. They perturbed
in calculating the average score of the movie, the user's average score, and the process of
constructing the covariance matrix. Then these perturbed matrices were published and used
in recommendation algorithms. Friedman et al. \cite{ABK16} introduced the differential privacy concept to
the matrix factorization recommendation algorithm. They mainly studied four different perturbation schemes
in the matrix factorization process: Input Perturbation, Stochastic Gradient Perturbation, Output Perturbation
and ALS with Output Perturbation. In contrast, the method of perturbation in matrix factorization proposed by
Friedman et al. \cite{ABK16} is similar to the method from McSherry et al. They all use
the Laplace noise to perform perturbation, and the results of the perturbation are clamped. However,
for such perturbation schemes, the sensitivity of the parameters controlling
the amount of noise is too large and has a great influence on the recommendation results. Zhu et al. \cite{ZRZ14}
applied differential privacy to a neighborhood-based collaborative filtering algorithm. They added Laplace noise
when calculating project or user similarity and used differential privacy for protection during neighbor selection.
Hua et al. \cite{HXZ15} proposed the objective function perturbation method. They perturbed the objective function
in the process of matrix factorization so that the perturbed item profile matrix satisfies the differential privacy.
Then the user profile matrix was published and used for user rating prediction.

\textbf{Personalized Differentially Private Recommendations.}
As for the differentially private recommendations mentioned above, users are usually provided with
the same level of privacy protection that does not care about  users' personal privacy needs. For this,
Li et al. \cite{LLW17} and Yang et al. \cite{YZX17}
adopted personalized differential privacy in recommendation systems.
They studied
user-based collaborative filtering recommendation algorithms. Indeed,
the algorithm is improved in \cite{LLW17} to make
recommended results more accurate. However, user-based collaborative filtering methods
cannot handle the user-item rating matrix with high sparsity, which reduces the accuracy of recommendation. Instead, we adopt the PMF method.

This work is based on perturbation of the objective function proposed by Hua et al. \cite{HXZ15}, and applies
the personalized differentially private sample mechanism proposed by Jorgensen et al. \cite{JYC15}. We make use of
the popular probabilistic matrix factorization model to design a personalized
differential privacy recommendation scheme PDP-PMF. We mention that
the sample mechanism proposed by Jorgensen et al. \cite{JYC15} cannot be directly applied here and
some adjustments are made later on. Therefore, in our scheme,
users can choose privacy protection levels for different projects accordingly,
while the amount of noise is reduced to enhance the accuracy of the recommendation results.



\section{Preliminaries}\label{sec:prelim}

In this section, we introduce some notations and initial definitions, and review
 probabilistic matrix factorization technique, (personalized)
differential privacy and the setting upon which our work is based.

\subsection{Probabilistic matrix factorization}

In our setting, there are $\textit{M}$ movies in the item
set $\mathcal{I}$ and $\textit{N}$ users in the user set $\mathcal{U}$.
Any user's rating for the movies is an integer in the range $[1,\textit{K}]$, where $(\textit{K}$ is a constant and called the upper limit of the score. We denote by $R=[r_{ij}]_{N\times M}$ the $N\times M$ preference matrix where the element
$r_{ij}$ represents the user $\tau_{i}$'s rating on the movie $\ell_{j}$. In general, $R$ is a
sparse matrix, which means that most users only rate a few movies relatively. The purpose of a recommendation system is
to predict the blank ratings in the user-item matrix $R$ and then
provide recommendation services for users.

As a probabilistic extension of the SVD model, probabilistic matrix factorization \cite{SM08} is a state-of-the-art technique to solve the above problem.
In the PMF model, the matrix $R$ would be given by the product of an $N\times d$ user coeficient matrix $U^T$ and a $d\times M$ item profile matrix $V$, where $U=[u_i]_{i\in [N]}$ and $V=[v_j]_{j\in [M]}$. These $d$-dimensional column vectors, $u_{i}$
and $v_{j}$, represent user-specific and item-specific latent profile vectors, respectively. Then the rating $r_{ij}$ is approximated by the product $u_i^T v_j$.
In general, the dimension $d$ is between $20$ and $100$.

Following literally \cite[Section 2]{SM08}, we add some Gaussian observation noise by defining a conditional distribution over the observed rating and placing zero-mean spherical Gaussian priors on all user and item profile vectors. The so-called maximization of the log-posterior over item and user profiles therein is equivalent to
the minimization of the sum-of-squared-errors objective function with quadratic regularization terms as follows,
\begin{equation}\label{eq:earmin}
\min_{U,V}E(U,V)\;=\;\frac{1}{2}\sum^{N}_{i=1}\sum^{M}_{j=1}I_{ij}
(r_{ij}-u_{i}^{T}v_{j})^{2}+\frac{\lambda_{u}}{2}\sum_{i=1}^{N}\|u_{i}\|_{2}^{2}
+\frac{\lambda_{v}}{2}\sum_{j=i}^{M}\|v_{j}\|_{2}^{2},
\end{equation}
where $\lambda_{u}>0$ and $\lambda_{v}>0$ are regularization parameters,
and $I_{ij}$ is an indicator function that is $I_{ij}=1$ when user $\tau_{i}$ rated item $\ell_{j}$ and otherwise $I_{ij}=0$. For each vector, we denote by $\|\cdot\|_{2}$
its Euclidean norm. As is mentioned in \cite{CM09,HXZ15}, it is reasonable to suppose that each column vector $u_i$ in $U$ satisfies $\|u_i\|_2\le 1$.

We use the stochastic gradient descent (SGD) method \cite{NIW13} to solve \eqref{eq:earmin}.
The profile matrices $U$ and $V$ are iteratively learned by the following rules:
\begin{equation}
  u_{i}(k)\;=\;u_{i}(k-1)-\gamma\cdot\nabla_{u_{i}}E(U(k-1),V(k-1)),
\end{equation}
\begin{equation}
  v_{j}(k)\;=\;v_{j}(k-1)-\gamma\cdot\nabla_{v_{j}}E(U(k-1),V(k-1)),
\end{equation}
where $\gamma>0$ is a learning rate parameter, $k$ is the number of iterations, and
\begin{equation}
  \nabla_{u_{i}}E(U,V)=-\sum_{j=1}^{M}I_{ij}(r_{ij}-u_{i}^{T}v_{j})v_{j}+\lambda_{u}u_{i},
\end{equation}
\begin{equation}
  \nabla_{v_{j}}E(U,V)=-\sum_{i=1}^{N}I_{ij}(r_{ij}-u_{i}^{T}v_{j})u_{i}+\lambda_{v}v_{j}.
\end{equation}
The initial $U(0)$ and $V(0)$ consist of uniformly random norm $1$ rows.

\subsection{Differential privacy}

Differential privacy \cite{Dwo08,DMN06,DR14} is a new privacy framework based on data distortion.
By adding controllable noise to the statistical results of the data,
DP guarantees that calculation results are not sensitive to any particular record in datasets.
In particular, adding, deleting or modifying any record in the dataset would not have significant influence on the results of statistical calculation based on the dataset.
Thus, as long as any adversary does not know the sensitive record in the dataset, the protection of the sensitive one can remain.

For applying DP, a crucial choice is the condition under which the datasets, $D$ and $D^\prime$, are considered to be neighboring.
In our scheme,  the user rating datasets, $D$ and $D^\prime$ are \emph{neighboring}
if $D$ can be obtained from $D^\prime$ by adding or removing one element.
Note that such a definition is usually adopted in \emph{Unbounded DP} \cite{LLS16}. In this paper, however, we mainly denote datasets by matrices with setting all unselected (and blank) ratings to zero. This yields that two neighboring datasets differ only at one record, in which place one (nonzero) rating is selected in a dataset (matrix) and unselected (zero) in the neighboring dataset (matrix). Due to such matrix representations, \emph{bounded DP} is applied herein and any replacement between nonzero ratings is not allowed.


\begin{defi}[Differential Privacy \cite{Dwo08}]\label{def:difpri}
A randomized algorithm $\mathcal{A}$ satisfies $\epsilon$-differential privacy, where $\epsilon\ge 0$, if and only if for
any two neighboring datasets $D$ and $D^\prime$, and for any subset $O$ of possible
outcomes in Range$(\mathcal{A})$,
\begin{equation}
  {\rm Pr}[\mathcal{A}(D)\in O]\leq \eu^{\epsilon} \cdot {\rm Pr}[\mathcal{A}(D^\prime)\in O].
\end{equation}
\end{defi}
\subsection{Personalized differential privacy}

In this section, we focus on the concept of \emph{Personalized Differential
Privacy (PDP)}. In Section 5 we present a new sample mechanism that satisfies the definition.

While in traditional DP the privacy guarantee is controlled by a uniform, global privacy
parameter (see Definition \ref{def:difpri}), PDP deploys a privacy
specification, in which each user-item pair in $\mathcal{U}\times \mathcal{I}$ independently specifies
the privacy requirement for their rating data. Formally,

\begin{defi}[Privacy Specification \cite{JYC15}]\label{def:prispec}
A privacy specification is a mapping $\mathcal{P}:\mathcal{U}\times \mathcal{I}\rightarrow \mathbb{R}_{+}^{N\times M}$
from user-item pairs to personal privacy preferences, where a smaller value represents a stronger privacy preference.
The notation $\epsilon_{ij}$ indicates the level of privacy specification of the user $\tau_{i}$ on the item $\ell_{j}$.
\end{defi}

In this paper, we use a matrix to describe a specific instance of a privacy specification,
e.g. $\mathcal{P}:=[\epsilon_{ij}]_{N\times M}$, where
$\epsilon_{ij}\in \mathbb{R}^{+}$. We assume that each user $\tau_{i}\in \mathcal{U}$ has an independent
level of privacy demand $\epsilon_{ij}$ for each item $\ell_{j}$, and that the default level of a relatively weak
privacy requirement is fixed as $\epsilon_{\rm def}=1.0$.
Indeed we are mainly intended to achieve $\epsilon_{ij}$-DP for each valid rating $r_{ij}$ while the $\epsilon_{ij}$ for those user-item pairs without valid ratings are meaningless.
  We are now ready to formalize our personalized privacy definition.

\begin{defi}[Personalized Differential Privacy \cite{JYC15}]\label{def:perdp}
In the context of a privacy specification $\mathcal{P}$ and and a universe of
user-item pairs $\mathcal{U}\times \mathcal{I}$,
an arbitrary random algorithm $\mathcal{A}:\mathcal{D}\rightarrow {\rm Range}(\mathcal{A})$ satisfies $\mathcal{P}$-personalized differential privacy ($\mathcal{P}$-PDP),
if for any two neighboring datasets $D, D^\prime\subseteq\mathcal{D}$ and all possible output sets $O\subseteq {\rm Range}(\mathcal{A})$,
\begin{equation}
  {\rm Pr}[\mathcal{A}(D)\in O]\leq \eu^{\epsilon_{ij}} \cdot {\rm Pr}[\mathcal{A}(D^\prime)\in O],
\end{equation}
where $\epsilon_{ij}$ denotes the level of privacy requirement of user $\tau_{i}$ for item $\ell_{j}$, also
referred to as privacy specification.
\end{defi}
\subsection{Setting}

In this paper, we consider a recommendation system associated with two kinds of actors,
users and the recommender. The recommender is assumed to be reliable,
which means that the recommender
does not abuse the user's rating data to gain profit.
Different users may have different privacy requirement for various items, which
means that users want to protect the rating data by individual privacy specification according to one's own will. As usual, each value of user's ratings is an integer between $1$ and $K$,
and even the value of $K$ is $5$ or $10$.

We are intended to design a personalized differentially private PMF scheme to meet each
user's personalized privacy requirement and ensure the accuracy of the
recommendation.
Moreover, we publish only
the perturbed item profile matrix $V$ but not the user profile matrix $U$ that must be kept confidential.
Otherwise, attackers can predict the user $\tau_i$'s preferences for all items locally by the user profile vector $u_i$ and the published matrix $V$.

\section{DP-PMF scheme}\label{sec:dpsch}


In this section, we design a differentially private recommendation scheme with a uniform privacy level $\epsilon$. This scheme then can serve as the basis of our scheme with personalized differential privacy in the next section.

In order to achieve these privacy goals, we apply the objective-perturbation
method proposed in \cite{HXZ15} to the probabilistic matrix factorization
technique, and perturb randomly the objective function instead of the output of the algorithm, to preserve the privacy of users. In our scenario, the perturbation of the objective
\eqref{eq:earmin} is as follows:
\begin{equation}\label{eq:pertobj}
  \min_{V}\widetilde{E}(V)=\frac{1}{2}\sum^{N}_{i=1}\sum^{M}_{j=1}I_{ij}(r_{ij}-u_{i}^{T}v_{j})^{2}
  +\frac{\lambda_{u}}{2}\sum^{N}_{i=1}\|u_{i}\|_{2}^{2}
  +\frac{\lambda_{v}}{2}\sum^{M}_{j=1}\|v_{j}\|_{2}^{2}
  +\sum^{M}_{j=1}\eta_{j}^{T}v_{j},
\end{equation}
with $Q=[\eta_{j}]_{j\in [M]}$ denoting the $d\times M$ noise matrix for the objective perturbation.

The detailed steps for the differentially private probabilistic matrix factorization
scheme are as follows:
\begin{enumerate}[Step 1:]
\item The recommender runs the original recommendation algorithm to solve
the objective function \eqref{eq:earmin} for obtaining the user profile matrix $U$
and store $U$ in private.
\item The user profile matrix is regarded as a constant of the perturbation
objective function \eqref{eq:pertobj} to achieve the item profile matrix
$\overline{V}$.
\end{enumerate}

In our scenario, the minimal valid rating $r_{\rm min}$ is usually a positive integer while those unselected and blank ratings are initially set to zero.

\begin{thm} \label{thm:DPperobj}
Suppose the noise vector $\eta_{j}\in Q$ in~\eqref{eq:pertobj}
is generated randomly and satisfies the probability density function
$P(\eta_{j})\propto e^{-\frac{\epsilon\|\eta_{j}\|_2}{\Delta}}$ with $\Delta=r_{\rm max}$.
Then, the item profile matrix $V$ derived from solving the objective function \eqref{eq:pertobj}
satisfies $\epsilon$-differential privacy.
\qed
\end{thm}
The proof can be finished in the same manner as in
the proof of Theorem 1 in \cite{HXZ15}.

\begin{proof}
Without loss of generality, we assume that $D=[r_{ij}]_{N\times M}$ and $D^\prime=[r^\prime_{ij}]_{N\times M}$
are two arbitrary neighboring datasets with only one score data different $r_{11}\neq r^\prime_{11}$. Let $Q=[\eta_{j}]_{j\in [M]}$
and $Q^\prime=[\eta^\prime_{j}]_{j\in [M]}$ be the noise matrices trained by~\eqref{eq:pertobj} using datasets $D$ and $D^\prime$, respectively.
Let $\overline{V}=[\bar{v}_{j}]_{j\in [M]}$ is the item profile matrix obtained
by minimizing~\eqref{eq:pertobj}.

We have, $\nabla_{v_{j}}\widetilde{E}(D,\bar{v}_{j})=\nabla_{v_{j}}\widetilde{E}(D^\prime,\bar{v}_{j})=0$
for any $j\in [M]$. Then
\begin{equation}\label{eq:pdao}
  \eta_{j}-\sum^{N}_{i=1}I_{ij}(r_{ij}-u_{i}^{T}\bar{v}_{j})u_{i}
  =\eta^\prime_{j}-\sum^{N}_{i=1}I_{ij}(r^\prime_{ij}-u_{i}^{T}\bar{v}_{j})u_{i}.
\end{equation}

If $j\neq 1$, we have $r_{ij}=r^\prime_{ij}$ and drive from \eqref{eq:pdao} that
$$\eta_{j}=\eta^\prime_{j}.$$
Hence, $\forall j\neq 1$, $\|\eta_{j}\|_2=\|\eta^\prime_{j}\|_2$.

If $j=1$ and $I_{11}=0$, we have by \eqref{eq:pdao} that
$$\eta_{j}=\eta^\prime_{j}.$$
Thus, $\|\eta_{j}\|_2=\|\eta^\prime_{j}\|_2$ when $j=1$ and $I_{ij}=0$.

If $j=1$ and $I_{11}=1$, we obtain that
$$
\eta_{j}-u_{1}(r_{11}-u^{T}_{1}\bar{v}_{1})=
\eta^\prime_{j}-u_{1}(r^\prime_{11}-u^{T}_{1}\bar{v}_{1}),
$$
\[
\eta_{j}-\eta^\prime_{j}=u_{1}(r_{11}-r^\prime_{11}).
\]
Thanks to $\|u_{1}\|_2\leq 1$ and $|r_{11}-r^\prime_{11}|\leq\Delta$, we obtain that $\|\eta_{j}-\eta^\prime_{j}\|_2\leq\Delta$.
Therefore,
\begin{align}\label{eq:enddpmf}
  \frac{{\rm Pr}[V=\bar{v}_{j}|D]}{{\rm Pr}[V=\bar{v}_{j}|D^\prime]}\;=\;&
  \frac{\prod_{j\in [M]}P(\eta_{j})}{\prod_{j\in [M]}P(\eta^\prime_{j})} \nonumber \\
    \;=\;&e^{\frac{-\epsilon \sum^{M}_{j=1}\|\eta_{j}\|_2+\epsilon \sum^{M}_{j=1}\|\eta^\prime_{j}\|_2}{\Delta}} \nonumber \\
    \;=\;&e^{\frac{-\epsilon(\|\eta_{1}\|_2-\|\eta^\prime_{1}\|_2)}{\Delta}} \nonumber \\
    \;\leq\;&e^{\epsilon}.
\end{align}
The above inequality is derived from the triangle inequality.

In the general case, two arbitrary neighboring datasets differing at the record $r_{pq}$ instead of $r_{11}$. Just a few trivial replacements are enough to prove the assertion as required.
\end{proof}

\section{PDP-PMF scheme}\label{sec:pdpsch}

This section is to propose a probabilistic matrix factorization recommendation scheme with personalized differential privacy.
With traditional differential privacy, DP-PMF provides only a uniform level of privacy protection for all users.
Our PDP-PMF scheme further improves upon the DP-PMF scheme to satisfy different users' privacy requirements and protect user privacy at the item level.

To accomplish our PDP-PMF scheme, we modify the sample mechanism in \cite{JYC15} to get a sample mechanism with bounded differential privacy.
The system carries out all processes on matrices including collection of raw dataset, random sampling, iterations, adding noises and the final release. In contrast to previous literatures, we adopt the concept of bounded DP instead of unbounded DP. While in recent years a few papers \cite{LLW17,YZX17} on PDP have appeared mainly for user-based collaborative filtering, we introduce PDP into PMF for preserving item-based privacy.

The new sample mechanism consists of two independent steps
basically as follows:


\begin{enumerate}[Step 1:]
\item Sample randomly the original rating data independently with some probability, set the
unselected rating data to zero, and output the sampling rating matrix.
\item
Apply the traditional (bounded) $t$-DP mechanism to protecting users' rating data (output above) privacy.
\end{enumerate}
\begin{defi}[Sample Mechanism with Bounded DP]\label{def:smpdp}
Given a recommendation mechanism $F$, a dataset $D$ and a privacy specification $\mathcal{P}$. Let $RSM(D,\mathcal{P},t)$ represent
the procedure that samples independently each rating $r_{ij}$ in $D$ with
probability $\pi(r_{ij},t)$ defined as
\begin{equation} \label{eq:smprob}
\pi(r_{ij},t)=\begin{cases}
\frac{e^{\epsilon_{ij}}-1}{e^{t}-1}\ \ \ {\rm if}\  t> \epsilon_{ij},\\
\ \ 1 \ \ \ \ \ \ \, {\rm otherwise},
\end{cases}
\end{equation}
where  ${\rm min}(\epsilon_{ij})\leq t\leq {\rm max}(\epsilon_{ij})$ is a configurable
threshold.
The Sample mechanism is defined as
\begin{equation}\label{eq:smeqdef}
  S_{F}(D,\mathcal{P},t)=DP^{t}_{F}(RSM(D,\mathcal{P},t)),
\end{equation}
where $DP^{t}_{F}$ is an arbitrary recommendation model $F$ satisfying
$t$-differential privacy.
\end{defi}

The procedure $DP^{t}_{F}$ is essentially treated as a black box that
operates on a sampled dataset and produces (perturbed) item profile matrix $V$.
 It could be a procedure satisfying differential privacy by applying a Laplace mechanism,
an exponential mechanism, or even a composition of several differentially private mechanisms.
For our scheme, we make use of DP-PMF model including objective perturbation.
From \cite{JYC15,LLW17}, the threshold
parameter $t$ provides a possible means of
balancing various types of error and is usually chosen as: $t={\rm max}(\epsilon_{ij})$ or
$t=\frac{1}{|\mathcal{R}^\prime|}\sum_{(i,j)\in\mathcal{R}^\prime}\epsilon_{ij}$, where $\mathcal{R}^\prime$ represents the training set of user-item pairs with available ratings. We will explore the best threshold
parameter by experiments in Sect. \ref{sec:exper}.


Our personalized differentially private recommendation scheme consists of three
stages $(S_{1},S_{2},S_{3})$, as shown in Figure \ref{fig:ourscheme}.
\begin{figure}[ht]
\centering
\includegraphics[scale=0.8]{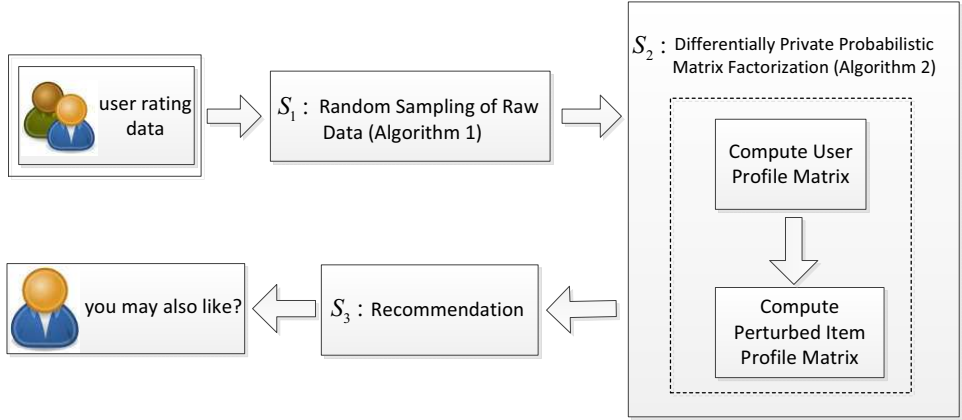}
\caption{PDP-PMF recommendation scheme}
\label{fig:ourscheme}
\end{figure}
\begin{enumerate}[(1)]
  \item \emph{$S_{1}$: Random sampling of raw data.}
We randomly sample the original dataset $D$ with a certain probability which is greatly dependent on each user-item based privacy demand $\epsilon_{ij}$ for each rating record $r_{ij}$'s sampling.
This stage outputs dataset matrix $D_{s}$,
see Algorithm \ref{alg:SampleRS}, where $RSM(D,\mathcal{P},t)$ indicates the randomized sampling algorithm.

\item \emph{$S_{2}$: (Traditional) differentially private probabilistic matrix factorization.}
Inspired by \cite{HXZ15}, we firstly run the PMF model \eqref{eq:earmin} on dataset $D_{s}$ output by stage $S_{1}$ to get the user profile matrix $U$ and store it secretly. Secondly, we operate
 the PMF model \eqref{eq:pertobj} with traditional $t$-DP to get the perturbed item profile matrix $\overline{V}$. Detailed steps are described in Algorithm \ref{alg:objobscure} for $DP^{t}_{F}(D_{s})$, where $I_{ij}$ is an indicator function on $D_s$ instead of $D$.

\item \emph{$S_{3}$: Recommendation.}
The system predicts users' ratings for items according to user profile matrix $U$ and item profile matrix $\overline{V}$
output by stage $S_{2}$ and provides recommendation services.
\end{enumerate}

\begin{algorithm}[tb] 
\caption{~~~$RSM(D,\mathcal{P},t)$} 
\label{alg:SampleRS}
\begin{algorithmic}[1] 
\REQUIRE
user set $\mathcal{U}$, item set $\mathcal{I}$, raw dataset (matrix) $D$,
privacy specification $\mathcal{P}$, sampling threshold $t$ 
\ENSURE  
sampled matrix $D_{s}$

\STATE Set $D_{s} \leftarrow D$
\FOR{each $i$ \ from\  $1$\  to\  $N$}
\FOR{each $j$ \ from\  $1$\  to\  $M$}
\IF{$D[i,j]\neq 0$}
\STATE Sample each rating pattern with probability \\ 
 $ \pi(D[i,j],t)=\begin{cases}
\frac{e^{\epsilon_{ij}}-1}{e^{t}-1}& \ \ \mbox{if} \ \ t> \epsilon_{ij}, \\
\ \ 1&\ \ \mbox{otherwise}.
\end{cases}$
\IF{$\pi(D[i,j],t)\neq1$}
\IF{$(\tau_{i},\ell_{j})$ is not selected}
\STATE $D_{s}[i,j]\leftarrow0$
\ENDIF
\ENDIF
\ENDIF
\ENDFOR
\ENDFOR
\RETURN $D_{s}$ 
\end{algorithmic}
\end{algorithm}

\begin{algorithm}[htbp]   
\caption{~~~$DP^{t}_{F}(D_{s})$}
\label{alg:objobscure}
\begin{algorithmic}[1]
\REQUIRE  
sampled profile set $D_{s}$, number of factors $d$, learning rate parameter $\gamma$,
regularization parameters $\lambda_{u}$ and $\lambda_{v}$, number of gradient
descent iterations $k_{1}$ and $k_{2}$, sampling threshold $t$
\ENSURE  
user profile matrix $U$, item profile matrix $\overline{V}$
\STATE Initialize random profile matrices $U$, $V$ and $\overline{V}$
\FOR{$\alpha$ \ from\  $1$\  to\  $k_{1}$}
\FOR{each $u_{i}$ and $v_j$}
\STATE
$
\nabla_{u_{i}}E\left(U(\alpha-1),V(\alpha-1)\right)\leftarrow$
\STATE
\quad \quad$-\sum_{j=1}^{M}I_{ij}
\left(D_s[i,j]-u^{T}_{i}(\alpha-1)v_{j}(\alpha-1)\right)v_{j}(\alpha-1)
+\lambda_{u}u_{i}(\alpha-1)
$
\STATE
$\nabla_{v_{j}}E(U(\alpha-1),V(\alpha-1))\leftarrow$
\STATE
\quad \quad$-\sum_{i=1}^{N}I_{ij}
\left(D_s[i,j]-u^{T}_{i}(\alpha-1)v_{j}(\alpha-1)\right)u_{i}(\alpha-1)+\lambda_{v}v_{j}(\alpha-1)$
\ENDFOR
\FOR{each $u_{i}$ and $v_j$}
\STATE $u_{i}(\alpha)\leftarrow u_{i}(\alpha-1)-\gamma\cdot\nabla_{u_{i}}E(U(\alpha-1),V(\alpha-1))$
\STATE $v_{j}(\alpha)\leftarrow v_{j}(\alpha-1)-\gamma\cdot\nabla_{v_{j}}E(U(\alpha-1),V(\alpha-1))$
\ENDFOR
\ENDFOR
\STATE Set $\overline{V}\leftarrow V$
\FOR{$\beta$ \ from\  $1$\  to\  $k_{2}$}
\FOR{$j$ \ from\  $1$\  to\  $M$}
\STATE Sample noise vector $\eta_{j}$ with probability density function $P(\eta_{j})\propto e^{\frac{-t\|\eta_{j}\|_2}{\Delta}}$
\STATE 
$\nabla_{\overline{v}_{j}}E(U,\overline{V}(\beta-1))\leftarrow
-\sum_{i=1}^{N}I_{ij}\left(D_s[i,j]-u^{T}_{i}\bar{v}_{j}(\beta-1)\right)u_{i}
+\lambda_{v}\bar{v}_{j}(\beta-1)+\eta_{j}$
\STATE $\bar{v}_{j}(\beta)\leftarrow \bar{v}_{j}(\beta-1)-\gamma\cdot\nabla_{\bar{v}_{j}}E(U,\overline{V}(\beta-1))$
\ENDFOR
\ENDFOR
\RETURN $U$ and $\overline{V}$ 
\end{algorithmic}
\end{algorithm}

\begin{thm}\label{thm:ourscheme}
The recommendation scheme proposed above satisfies $\mathcal{P}$-personalized differential privacy.
\qed
\end{thm}
\begin{proof}
Let $D=[r_{ij}]_{N\times M}$ and $D^\prime=[r^\prime_{ij}]_{N\times M}$ denote two (neighboring)  sets of ratings
different only at one record, $r_{pq}\neq 0$ and $r^\prime_{pq}=0$.
This means that the dateset $D^\prime$ results from removing from $D$ the rating $r_{pq}$.
We will establish that
for any output $O\in {\rm Range}(S_{F}(D,\mathcal{P},t))$,
\begin{equation}\label{eq:pdppmf}
e^{-\epsilon_{pq}}\cdot {\rm Pr}[S_{F}(D^\prime,\mathcal{P},t)\in O] \leq {\rm Pr}[S_{F}(D,\mathcal{P},t)\in O]\leq e^{\epsilon_{pq}}\cdot {\rm Pr}[S_{F}(D^\prime,\mathcal{P},t)\in O].
\end{equation}

We first prove the right hand of \eqref{eq:pdppmf}.
Note that all possible outputs of $RSM(D,\mathcal{P},t)$, see Algorithm \ref{alg:SampleRS}, can be divided into two parts.
One case is that the score $r_{pq}$ is selected and the other case unselected. Still, we denote the output set by a matrix
$Z=[r^{\prime\prime}_{ij}]_{N\times M}$. For convenience, the notation $Z\preceq D^\prime$ means that for any $1\le i\le N$ and $1\le j\le M$,\ the relation, $ r^{\prime\prime}_{ij}=r^\prime_{ij}$ or $r^{\prime\prime}_{ij}=0$, holds. We denote by $Z_{+r_{pq}}$ the dataset matrix differing from $Z$ at only one (additional) record $r_{pq}\neq  0$.
Then we rewrite \eqref{eq:pdppmf} as follows,
\begin{align}\label{eq:schemepro1}
  {\rm Pr}[S_{F}(D,\mathcal{P},t)\in O] \,&=\, \sum_{Z\preceq D^\prime}\Big(\pi(r_{pq},t)\cdot{\rm Pr}[RSM(D^\prime,\mathcal{P},t)=Z]\cdot{\rm Pr}[DP^{t}_{F}(Z_{+r_{pq}})\in O]\Big) \nonumber \\
&\ \ \ \,+\, \sum_{Z\preceq D^\prime}\Big((1-\pi(r_{pq},t))\cdot{\rm Pr}[RSM(D^\prime,\mathcal{P},t)=Z]\cdot{\rm Pr}[DP^{t}_{F}(Z)\in O]\Big)  \nonumber \\
&\,=\, \sum_{Z\preceq D^\prime}\Big(\pi(r_{pq},t)\cdot{\rm Pr}[RSM(D^\prime,\mathcal{P},t)=Z]\cdot{\rm Pr}[DP^{t}_{F}(Z_{+r_{pq}})\in O]\Big)    \nonumber \\
&\ \ \ \,+\, (1-\pi(r_{pq},t))\cdot{\rm Pr}[S_{F}(D^\prime,\mathcal{P},t)\in O].
\end{align}
From  \eqref{eq:smeqdef}, $DP^{t}_{F}$ satisfies $t$-differential privacy. Then we have, further,

\begin{align}\label{eq:schemepro2}
  {\rm Pr}[S_{F}(D,\mathcal{P},t)\in O] \,&\leq\, \sum_{Z\preceq D^\prime}\Big(\pi(r_{pq},t)\cdot{\rm Pr}[RSM(D^\prime,\mathcal{P},t)=Z]\cdot e^{t}\cdot{\rm Pr}[DP^{t}_{F}(Z)\in O]\Big) \nonumber \\
&\ \ \ \,+\, (1-\pi(r_{pq},t))\cdot{\rm Pr}[S_{F}(D^\prime,\mathcal{P},t)\in O]\nonumber \\
&\,=\, e^{t}\cdot\pi(r_{pq},t)\cdot{\rm Pr}[S_{F}(D^\prime,\mathcal{P},t)\in O] \,+\, \left(1-\pi(r_{pq},t)\right)\cdot{\rm Pr}[S_{F}(D^\prime,\mathcal{P},t)\in O] \nonumber \\
&\,=\, \left(1-\pi(r_{pq},t)+e^{t}\cdot\pi(r_{pq},t)\right)\cdot{\rm Pr}[S_{F}(D^\prime,\mathcal{P},t)\in O].
\end{align}
Now in view of \eqref{eq:smprob} we consider two cases for $\pi(r_{pq},t)$ as below.

\begin{enumerate}[(a)]
  \item The case $t\leq\epsilon_{pq}$. The rating record $r_{pq}$ is selected with probability $\pi(r_{pq},t)=1$. This implies
\[\begin{split}
  {\rm Pr}[S_{F}(D,\mathcal{P},t)\in O] \,&\leq\, \left(1-\pi(r_{pq},t)+e^{t}\cdot\pi(r_{pq},t)\right)\cdot{\rm Pr}[S_{F}(D^\prime,\mathcal{P},t)\in O] \\
&\,=\, e^{t}\cdot{\rm Pr}[S_{F}(D^\prime,\mathcal{P},t)\in O] \\
&\,\leq\, e^{\epsilon_{pq}}\cdot{\rm Pr}[S_{F}(D^\prime,\mathcal{P},t)\in O].
\end{split}\]
  \item The case $t>\epsilon_{pq}$. The rating record $r_{pq}$ is selected with probability $\pi(r_{pq},t)=\frac{e^{\epsilon_{pq}}-1}{e^{t}-1}$. Then we obtain
\[
1-\pi(r_{pq},t)+e^{t}\cdot\pi(r_{pq},t)=1-\frac{e^{\epsilon_{pq}}-1}{e^{t}-1}+e^{t}\cdot\frac{e^{\epsilon_{pq}}-1}{e^{t}-1}
=e^{\epsilon_{pq}},
\]
as required. This finishes our proof of the right inequality of \eqref{eq:pdppmf}.
\end{enumerate}

The proof of the left inequality of \eqref{eq:pdppmf} can follow the same line as above with a few modifications. Firstly, using
\eqref{eq:schemepro1} and $t$-differential privacy from $DP^{t}_{F}$, we observe that
\begin{equation}\label{eq:schemepro3}
  {\rm Pr}[S_{F}(D,\mathcal{P},t)\in O] \,\ge\,
 \left(1-\pi(r_{pq},t)+e^{-t}\cdot\pi(r_{pq},t)\right)\cdot{\rm Pr}[S_{F}(D^\prime,\mathcal{P},t)\in O].
\end{equation}

Secondly, while the case of $t\leq\epsilon_{pq}$ is trivial, we mention for the case of $t>\epsilon_{pq}$ that
\[
\begin{split}
  1-\pi(r_{pq},t)+e^{-t}\cdot\pi(r_{pq},t) \,&=\, 1+\frac{e^{\epsilon_{pq}}-1}{e^{t}-1}\cdot\left(e^{-t}-1\right)
=1-e^{-t} \cdot \left(e^{\epsilon_{pq}}-1 \right) \\
&\,>\,1-e^{-\epsilon_{pq}} \cdot \left(e^{\epsilon_{pq}}-1 \right)=e^{-\epsilon_{pq}}.
\end{split}
\]

Finally, the proof is complete as required.
\end{proof}

\begin{rem}
Theorem \ref{thm:ourscheme} is our main theoretic assertion. Indeed, we would like also to consider the general choice of modification (between nonzero ratings) for defining neighboring matrices \cite{HR12}. That is, two neighboring matrices (sets) of ratings
differ only at one (non-zero) record $(r_{pq} \neq r^\prime_{pq})$, which allows $r_{pq}, r^\prime_{pq}\neq 0$.
In this general \emph{bounded DP} case, Theorem \ref{thm:DPperobj} holds still with the sensitivity $\Delta=r_{\rm max}-r_{\rm min}$ instead of $r_{\rm max}$. As for Theorem \ref{thm:ourscheme}, the recommendation scheme proposed preserves obviously $\mathcal{P}^\prime$-personalized differential privacy where $\mathcal{P}^\prime:=[\epsilon_{ij}^\prime]_{N\times M}$ and $\epsilon_{ij}^\prime = 2 \epsilon_{ij}$. It is a pity that we do not achieve $\epsilon_{ij}^\prime =  \epsilon_{ij}$ theoretically.

\end{rem}

\section{Experimental evaluation}\label{sec:exper}

In this section, we mainly present an experimental evaluation of our PDP-PMF scheme.
In particular, we compare the recommendation quality of PDP-PMF with that of DP-PMF.

\subsection{Experimental Setup}
\textbf{Experimental comparison.}
Probabilistic matrix factorization technique has been shown to provide relatively high
predictive accuracy  \cite{SM08}. In order to verify the recommendation quality of the proposed PDP-PMF scheme, we carry out
a series of experiments to compare it with the DP-PMF scheme.
From \cite{Kor10, MYL08}, it is known that the recommended accuracy of PMF is higher than that of
user-based collaborative filtering method. 
For this reason, we do not make comparisons with the differentially private recommendation schemes based on collaborative filtering method from \cite{LLW17,YZX17}.

\textbf{Selection of Datasets.} We evaluate the schemes with three public datasets: MovieLens 100K dataset (ML-100K)\footnote{https://grouplens.org/datasets/movielens/\label{web:dataset1}},
MovieLens 1M dataset (ML-1M)\textsuperscript{\ref{web:dataset1}} and Netflix dataset\footnote{https://www.netflixprize.com/}.
To be specific,
the MovieLens $100$K dataset contains $100$ thousand ratings of $1,682$ movies from $943$ users.
The MovieLens $1$M dataset contains $6,000$ users' ratings of $4,000$ movies and a total of $1$ million records.
The third dataset is reduced randomly from
the famous Netflix dataset and consists of $1.5$ million ratings given by
$7,128$ users over 16,907 movies. For the above datasets,
each rating value is an integer in $[1, 5]$.

\textbf{Evaluation Metrics.}
We divide the available rating data into the training and testing
sets and use tenfold cross validation to train and evaluate the recommendation system.
In our experiments, we focus on the prediction accuracy of the recommendations.
The root mean square error (RMSE) is defined by
${\rm RMSE}=\sqrt{\sum_{(i,j)\in\mathcal{R}}(\hat{r}_{i,j}-r_{ij})^{2}/{|\mathcal{R}|}}$,
where $\hat{r}_{ij}=u^{T}_{i}\bar{v}_{j}$ is the predicted user $\tau_{i}$'s rating of movie $\ell_{j}$,
and $r_{ij}$ the true rating,
$\mathcal{R}$ represents the testing set of ratings being predicted, and $|\mathcal{R}|$ denotes the number of test ratings.
The cumulative distribution function (CDF) computes the percentage of ratings whose
prediction errors, i.e.,
$\{|\hat{r}_{ij}-r_{ij}|:\,(i,j)\in\mathcal{R}\}$, are less than or equal to the variable ranging in $[0, r_{\rm max}]$.


\textbf{Parameter Selection.} In the PDP-PMF and DP-PMF schemes, some specific parameters are fixed as follows:
learning rate $\gamma=50$, regularization parameters $\lambda_{u}=\lambda_{v}=0.01$,
The maximum number of iterations $k_1=50$ and dimension of all the profile matrix $d=20$.
In our proposed PDP-PMF scheme, there are two more important parameters, privacy specification
$\mathcal{P}$ and
sampling threshold $t$.

In fact, some studies in literatures \cite{AG05,BGS05} related to psychology have found that users can be divided into different
groups based on levels of privacy needs.
In the experiment we randomly divided all records for movie ratings
into three levels: \emph{conservative} indicates that the records have high
privacy requirements among $[\epsilon_{c},\epsilon_{m})$; \emph{moderate} means that the records'
privacy concern is at a middle level ranging in $[\epsilon_{m},\epsilon_{l})$; \emph{liberal} represents
the records with low privacy concern $\epsilon_{l}$. The percentages of the ratings for
\emph{conservative} and \emph{moderate} security levels are denoted by $f_{c}$ and $f_{m}$, respectively.
Therefore, the proportion of the low-level privacy requirement level \emph{liberal} is $f_{l}=1-f_{c}-f_{m}$.

Based on findings reported in \cite{AG05} in the context of a
user survey regarding privacy attitudes, the default fractions of user-item records in the  \emph{conservative}
and \emph{moderate} groups in our experiment are set to $f_{c}=0.54$ and $f_{m}=0.37$, respectively.
 In addition, the privacy preferences
at the \emph{conservative} and \emph{moderate} levels are picked uniformly at random from the ranges
$[\epsilon_{c},\epsilon_{m})$ and $[\epsilon_{m},\epsilon_{l})$, respectively. As usual, we specify
the privacy preference of the low-level privacy requirement \emph{liberal} $\epsilon_{l}=1.0$, the default values of
$\epsilon_{c}$, $\epsilon_{m}$ are $\epsilon_{c}=0.1$ and $\epsilon_{m}=0.2$, respectively, 
and the sampling threshold
$t=\frac{1}{|\mathcal{R}^\prime|}\sum_{(i,j)\in\mathcal{R}^\prime}\epsilon_{ij}$
(i.e., the average privacy setting).
To meet the privacy requirements of all users, the privacy parameter of the DP-PMF scheme is set to the minimum
privacy preference, i.e., $\epsilon=\epsilon_{c}=0.1$, where a smaller value
implies greater privacy.

Tables \ref{Tab:paraSetf2}-\ref{Tab:paraSetf5} list the various parameters used
in our experiments.
For valuable comparisons, we mainly consider the ways of changing the privacy specification ($f_c$ or $\epsilon_m$) and the sampling threshold, respectively, as follows.
\begin{table}[htbp]
    \begin{spacing}{1.5} 
    \caption{Parameter sets for Fig. \ref{fig:subfigfc}}\label{Tab:paraSetf2}
    \end{spacing}
    \centering
    \begin{tabular}{|l|c|r|} 
        \hline
            Parameter set& Parameter value\\  
        \hline
            \multirow{3}*{PDP-PMF}& \multirow{3}*{ \makecell[tl]{$f_{m}=0.37$, $f_{l}=1-f_{c}-f_{m}$, $\epsilon_{c}=0.1$,\\ $\epsilon_{m}=0.2$,
            $\epsilon_{l}=1.0$, $t=\frac{1}{|\mathcal{R}^\prime|}\sum_{(i,j)\in\mathcal{R}^\prime}\epsilon_{ij}$} }\\
            ~& ~\\
            ~& ~\\
        \hline
            DP-PMF& $\epsilon=0.1$\\
        \hline
    \end{tabular}
\end{table}
\begin{figure}[htbp]  
  \centering
  \subfigure[ML-100K]{
    \label{fig:subfigfc100k:a} 
    \includegraphics[width=2.77in]{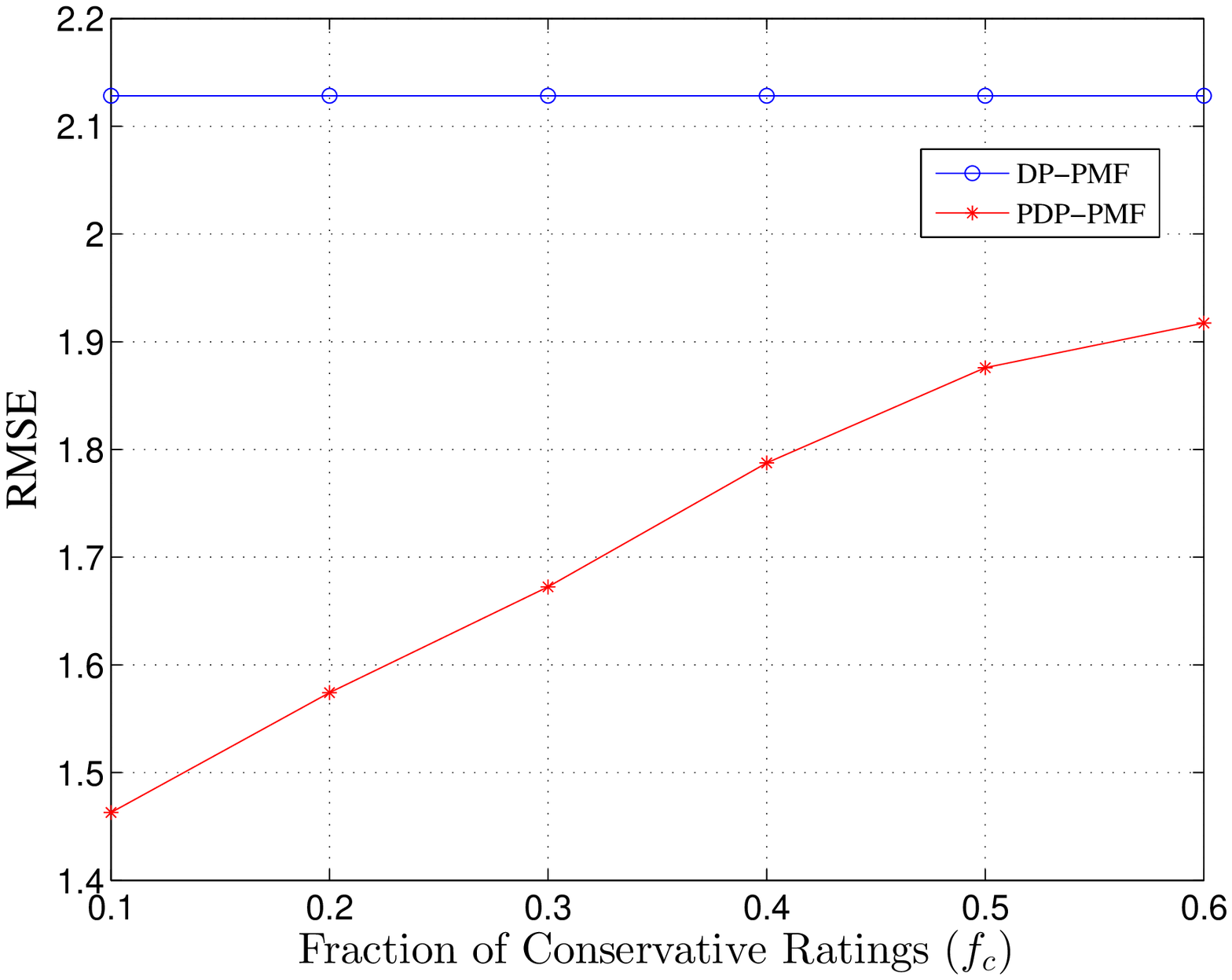}
  }
  \subfigure[ML-1M]{
    \label{fig:subfigfc1m:b} 
    \includegraphics[width=2.77in]{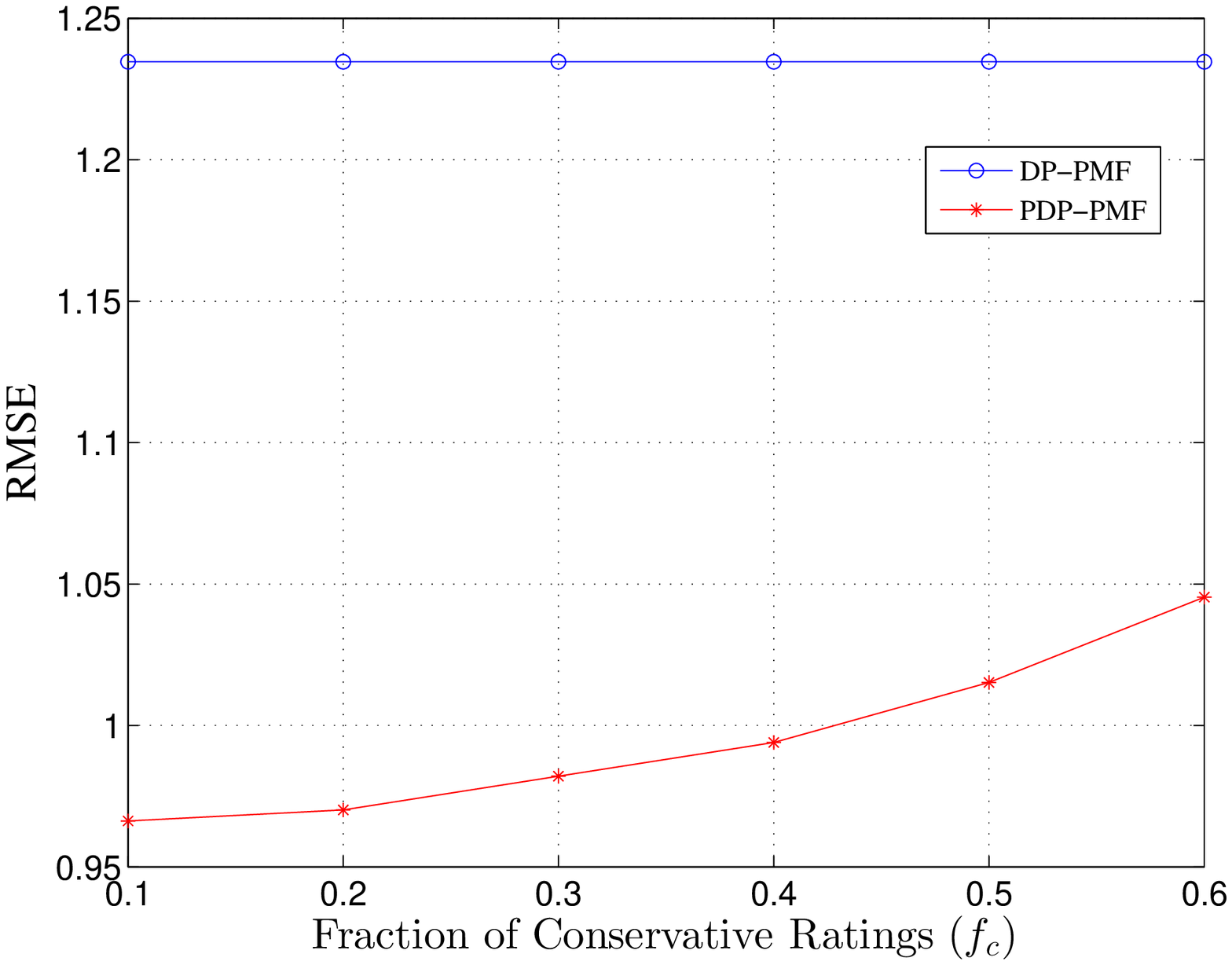}
  }
  \subfigure[Netflix]{
    \label{fig:subfigfcnet:c} 
    \includegraphics[width=2.77in]{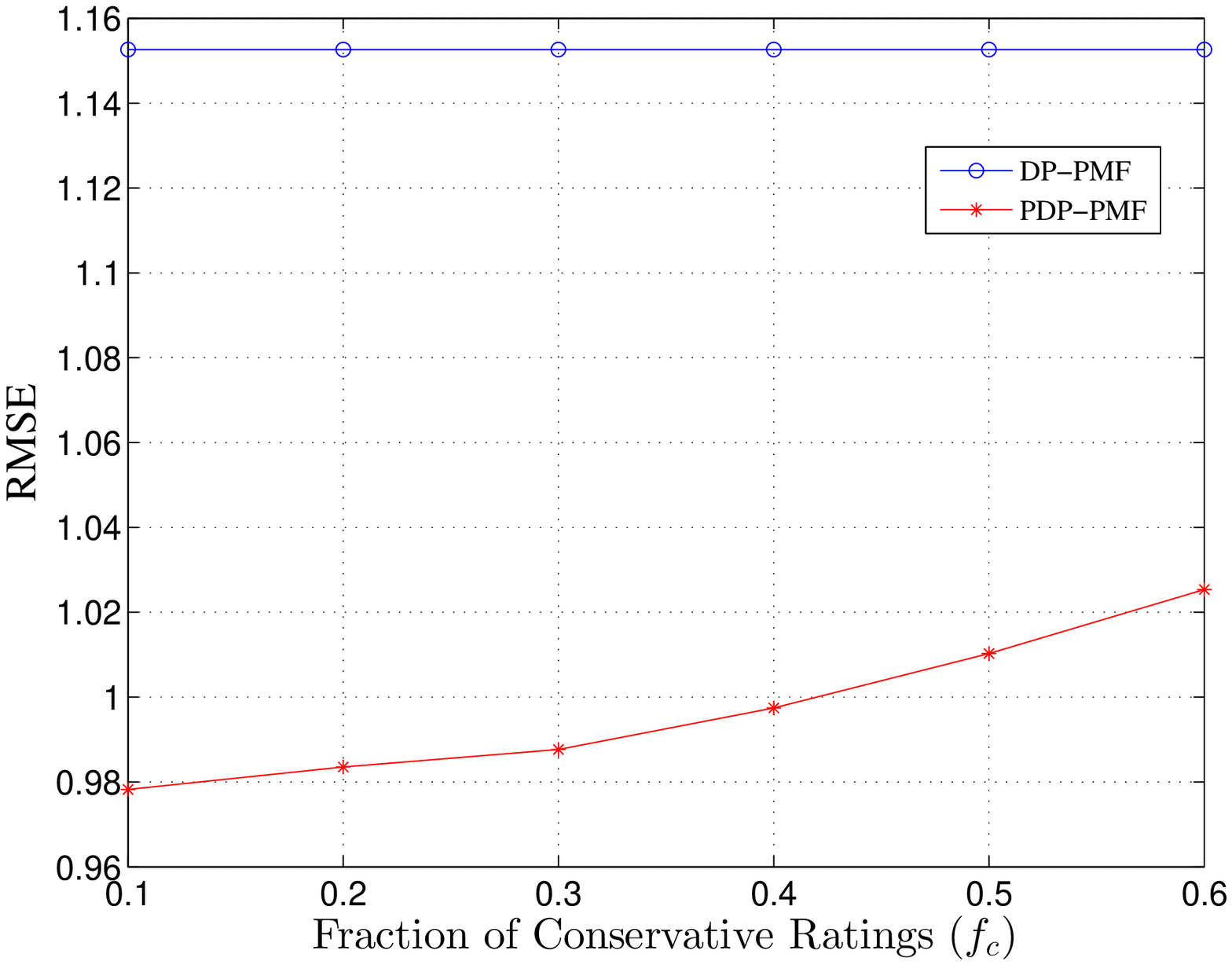}
  }
  \caption{RMSE of each scheme for prediction task, as $f_{c}$ is varied.
  }
  \label{fig:subfigfc} 
\end{figure}
\begin{table}[htbp]
    \begin{spacing}{1.5} 
    \caption{Parameter sets for Fig. \ref{fig:subfigcdf}}\label{Tab:paraSetf4}
    \end{spacing}
    \centering
    \begin{tabular}{|l|c|r|} 
        \hline
            Parameter set& \multicolumn{2}{c|}{Parameter value}\\  
        \hline
            PDP-PMF1& $f_{c}=0.54$, $f_{l}=0.09$& \multirow{3}*{ \makecell[tl]{$f_{m}=0.37$, $\epsilon_{c}=0.1$, $\epsilon_{m}=0.2$, \\
             $\epsilon_{l}=1.0$, $t=\frac{1}{|\mathcal{R}^\prime|}\sum_{(i,j)\in\mathcal{R}^\prime}\epsilon_{ij}$} } \\
        \cline{1-2}
            PDP-PMF2& $f_{c}=0.37$, $f_{l}=0.26$& ~\\
        \cline{1-2}
            PDP-PMF3& $f_{c}=0.20$, $f_{l}=0.43$& ~\\
        \hline
            DP-PMF& \multicolumn{2}{c|}{$\epsilon=0.1$}\\
        \hline
    \end{tabular}

\end{table}
\begin{figure}[htbp]  
  \centering
  \subfigure[ML-100K]{
    \label{fig:subfigcdf100k:a} 
    \includegraphics[width=2.73in]{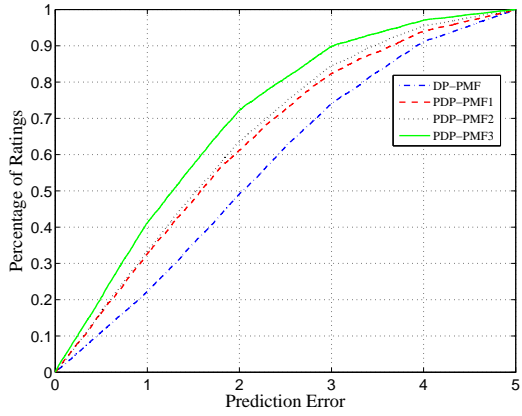}
  }
  \subfigure[ML-1M]{
    \label{fig:subfigcdf1m:b} 
    \includegraphics[width=2.73in]{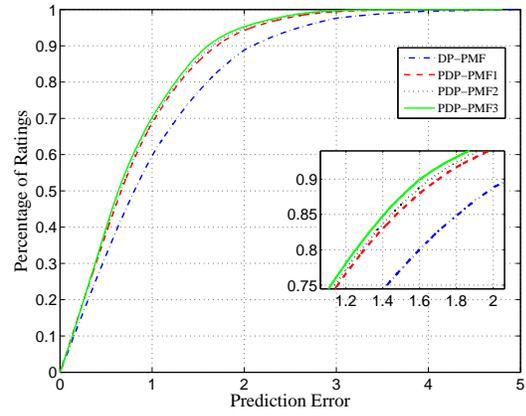}
  }
  \subfigure[Netflix]{
    \label{fig:subfigcdfnet:c} 
    \includegraphics[width=2.73in]{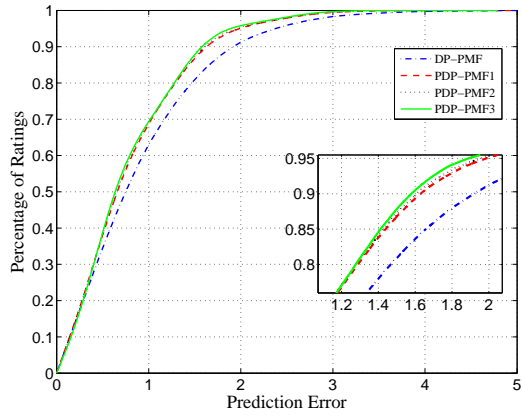}
  }
  \caption{CDF of prediction errors in each scheme.
  }
  \label{fig:subfigcdf} 
\end{figure}
\begin{table}[htbp]
    \begin{spacing}{1.5} 
    \caption{Parameter sets for Fig. \ref{fig:subfigem}} \label{Tab:paraSetf3}
    \end{spacing}
    \centering
    \begin{tabular}{|l|c|r|} 
        \hline
            Parameter set& \multicolumn{2}{c|}{Parameter value}\\  
        \hline
            PDP-PMF1& $f_{c}=0.54$, $f_{l}=0.09$& \multirow{3}*{ \makecell[tl]{$f_{m}=0.37$, $\epsilon_{c}=0.1$, $\epsilon_{l}=1.0$, \\
           $t=\frac{1}{|\mathcal{R}^\prime|}\sum_{(i,j)\in\mathcal{R}^\prime}\epsilon_{ij}$} } \\
        \cline{1-2}
            PDP-PMF2& $f_{c}=0.37$, $f_{l}=0.26$& ~\\
        \cline{1-2}
            PDP-PMF3& $f_{c}=0.20$, $f_{l}=0.43$& ~\\
        \hline
            DP-PMF& \multicolumn{2}{c|}{$\epsilon=0.1$}\\
        \hline
    \end{tabular}

\end{table}
\begin{figure}[htbp]  
  \centering
  \subfigure[ML-100K]{
    \label{fig:subfigem100k:a} 
    \includegraphics[width=2.77in]{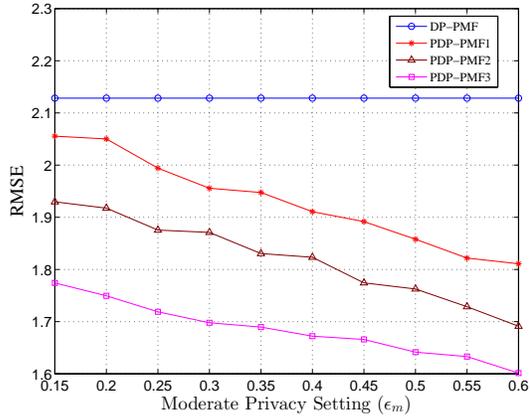}
  }
  \subfigure[ML-1M]{
    \label{fig:subfigem1m:b} 
    \includegraphics[width=2.77in]{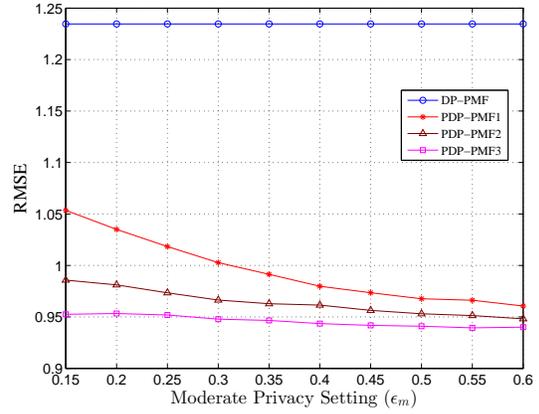}
  }
  \subfigure[Netflix]{
    \label{fig:subfigemnet:c} 
    \includegraphics[width=2.77in]{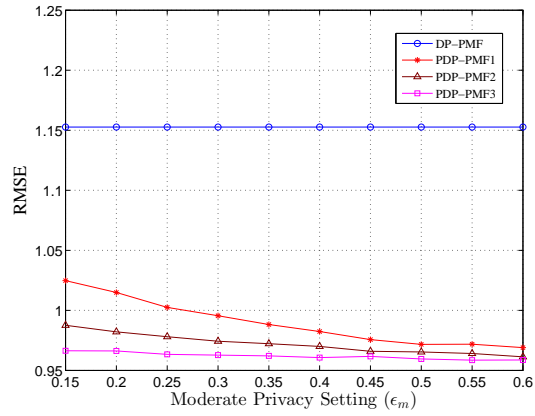}
  }
  \caption{RMSE of each scheme for prediction task, as $\epsilon_{m}$ is varied.
  }
  \label{fig:subfigem} 
\end{figure}
\begin{table}[htbp]
    \begin{spacing}{1.5} 
    \caption{Parameter sets for Fig. \ref{fig:subfig-t}}\label{Tab:paraSetf5}
    \end{spacing}
    \centering
    \begin{tabular}{|l|c|r|} 
        \hline
            Parameter set& \multicolumn{2}{c|}{Parameter value}\\  
        \hline
            PP-t1& $t=0.6$& \multirow{3}*{ \makecell[tl]{\vspace{-2mm}\\ $f_{c}=0.60$,\\ $f_{m}=0.35$, \\ $f_{l}=0.05$,\\ $\epsilon_{c}=0.1$,\\ $\epsilon_{m}=0.4$,\\           $\epsilon_{l}=1.0$ }
            } \\
        \cline{1-2}
            PP-t2& $t=0.7$& ~\\
        \cline{1-2}
            PP-t3& $t=0.8$& ~\\
        \cline{1-2}
            PP-t4& $t=1.0$& ~\\
            \cline{1-2}
            PET1& ${\rm PET}1(t)={\rm CDF}(1)$\ {\rm with}\ $t$& ~\\
        \cline{1-2}
            PET1h&{\rm PET1h}$(t)={\rm CDF}(1.5)$ {\rm with}\ $t$& ~\\
        \cline{1-2}
            PET2& ${\rm PET}2(t)={\rm CDF}(2)$\ {\rm with}\ $t$& ~\\
                \hline
    \end{tabular}

\end{table}
\begin{figure}[htbp]  
  \centering
  \subfigure[CDF of prediction errors in PDP-PMF]{
    \label{fig:subfigcdf100k:a} 
    \includegraphics[width=2.73in]{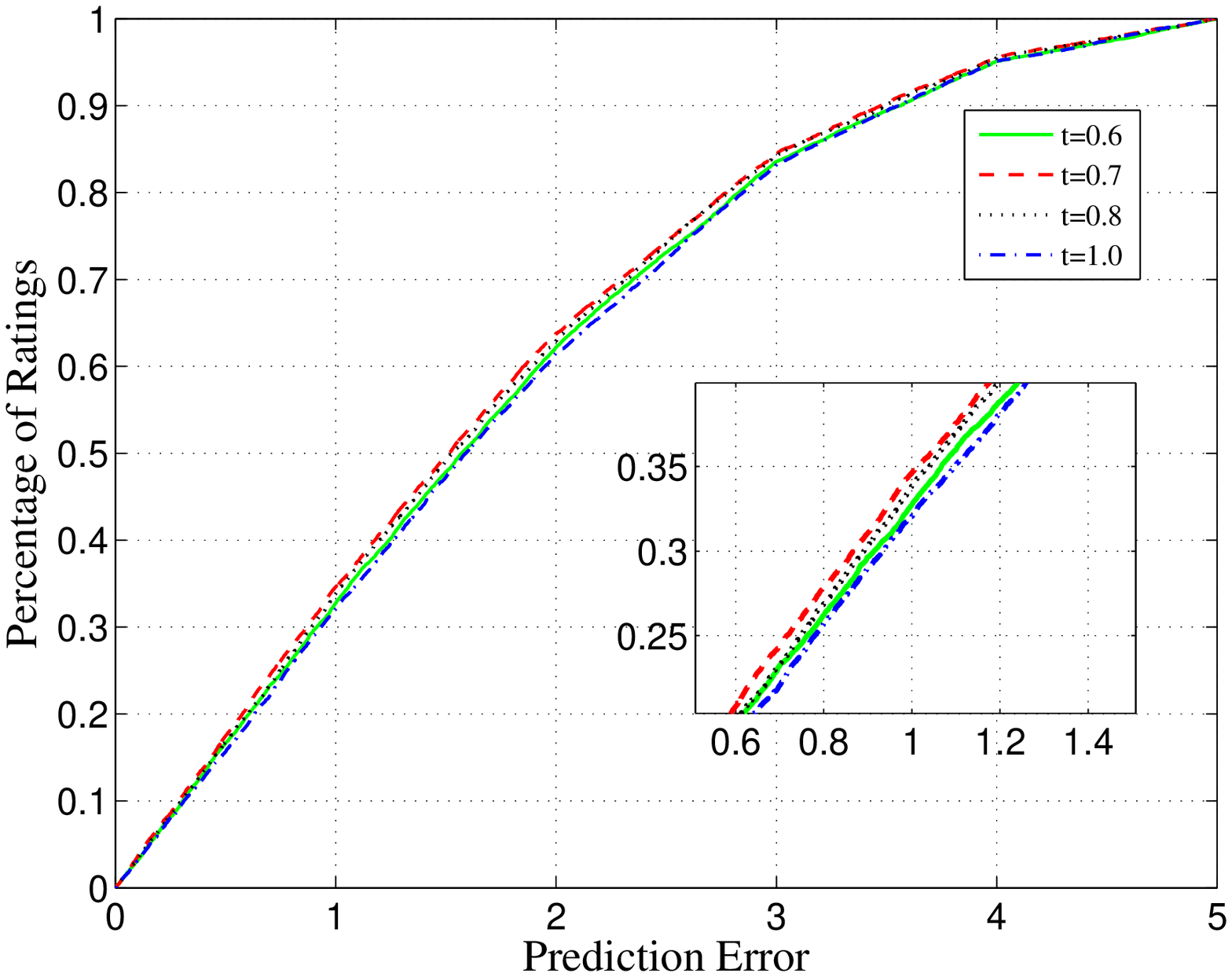}
  }
  \subfigure[CDF at three error points, $1.0, 1.5, 2.0$]{
    \label{fig:subfigcdf100k:b} 
    \includegraphics[width=2.73in]{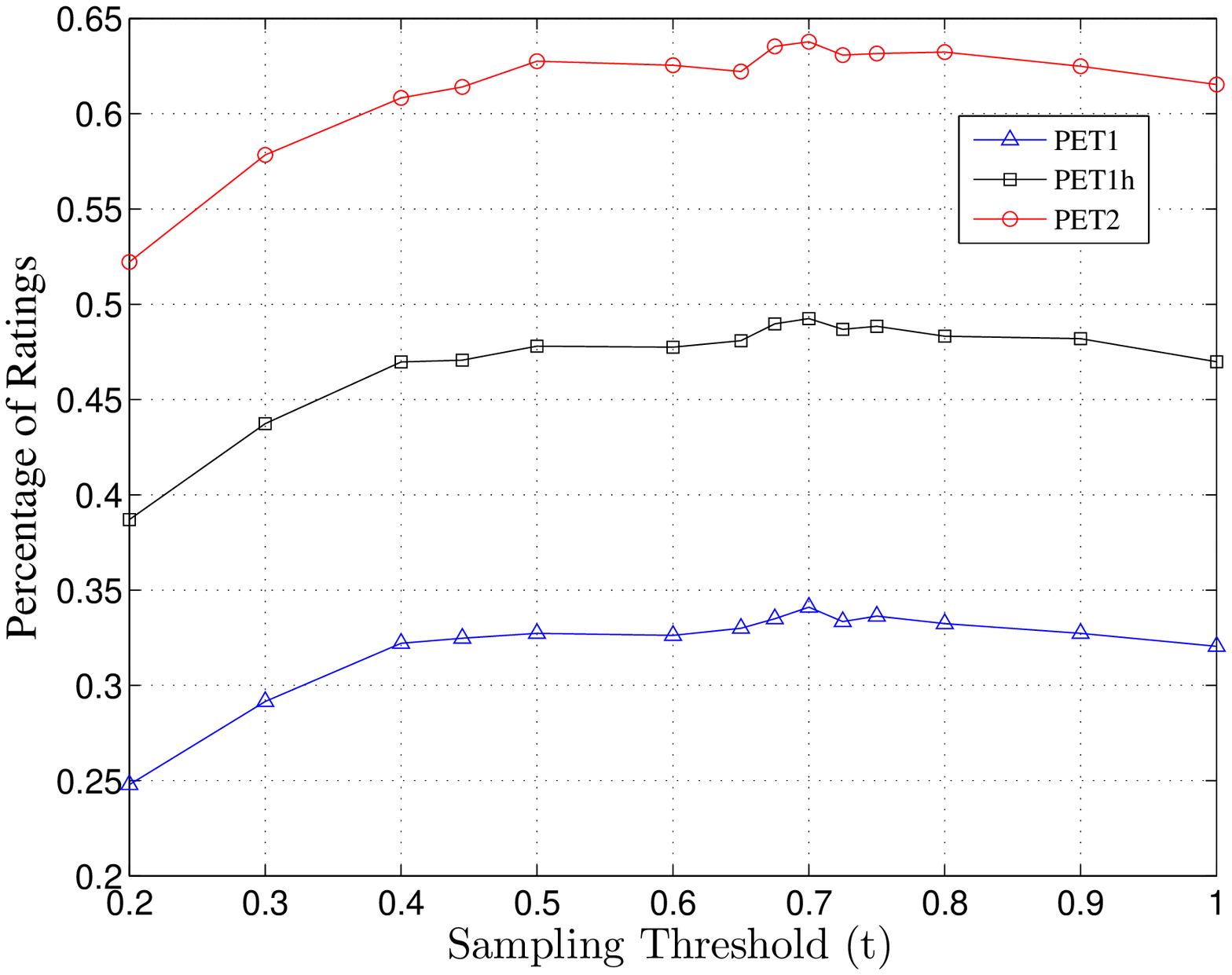}
  }

  \caption{CDF of prediction errors for ML-100K in PDP-PMF, as $t$ is varied.
  }\label{fig:subfig-t} 
\end{figure}

\subsection{Results}
{\bf Varying the Privacy Specification.}
We perform this firstly by varying $f_c$ (from $0.1$ to $0.6$), while
keeping the other parameters at their defaults. For
each setting of $f_c$ with $f_{m}$ fixed, the fraction of liberal users is equal to $1-f_{c}-f_{m}$, then decreasing $f_c$ increases the number of liberal
users, see Tables \ref{Tab:paraSetf2} and \ref{Tab:paraSetf4}.
We expect that the PDP-PMF
performs even better when the liberal fraction $f_l$ is greater,
cf. Figs. \ref{fig:subfigfc} and \ref{fig:subfigcdf}.

Figure \ref{fig:subfigfc} describes RMSEs of two schemes for prediction task with varied $f_{c}$. In general,
RMSE of PDF-PMF scheme increases approximately linearly with the increasing of fraction of conservative records $f_c$ and is much smaller than that of DP-PMF scheme (staying unchanged).  For
Figs. \ref{fig:subfigfc1m:b} and \ref{fig:subfigfcnet:c}, the RMSEs are always around $1.0$.

Figure \ref{fig:subfigcdf} illustrates CDF of prediction errors for four schemes with different $f_{c}$. We observe that higher $f_{c}$ makes the average privacy preference $t$ smaller, results in more noises and gives worse prediction.
For Figs. \ref{fig:subfigcdf1m:b} and \ref{fig:subfigcdfnet:c}, the curves for the three PDP-PMF schemes
extremely overlap each other and their percentages of ratings with prediction errors less than $1$ achieve $70\%$.

Next, we vary $\epsilon_{m}$, which controls the upper (lower) bound on
the range of privacy settings generated for the conservative (moderate)
user-item records, see Table \ref{Tab:paraSetf3}.
We expect the error for the PDP-PMF approaches
to be smaller with a higher $\epsilon_{m}$,
cf. Fig. \ref{fig:subfigem}.

Figure \ref{fig:subfigem} demonstrates RMSEs of four schemes for prediction task with varied $\epsilon_{m}$. Similarly, RMSEs of PDF-PMF schemes decreases approximately linearly with the increasing of moderate privacy setting $\epsilon_{m}$ and is much smaller than that of DP-PMF scheme (unchanged). For
Figs. \ref{fig:subfigem1m:b} and \ref{fig:subfigemnet:c}, the RMSEs
even slightly deceases compared with Fig. \ref{fig:subfigem100k:a},
 and they approach some values in trend, respectively, among $[0.94,0.97]$, and in particular, the RMSEs of PDP-PMF3 are almost flat.

{\bf Searching for the best sampling threshold.} We vary the sampling threshold $t$, which balances various types of error, see Table \ref{Tab:paraSetf5}.  We would like to find the best sampling threshold and check whether the simply setting $t = \epsilon_l= 1.0$ offers good results in our setting, cf. Fig. \ref{fig:subfig-t}.



Figure \ref{fig:subfig-t} shows that the sampling threshold has  a slight influence on the recommendation quality.
Figure \ref{fig:subfigcdf100k:b} helps us to find the best choice of sampling threshold $t\approx 0.70$ for the given setting, and it demonstrates that
when $t>0.4$ (including the theoretically average privacy preference $0.445$), the variations on the percentages of ratings are quite small. In particular, the difference between the percentages of ratings at the approximately maximum point $0.7$ and the liberal point $1.0$ is around $2.2\%$ for each of the three curves.

Finally, from the above experiments, we observe that our proposed PDP-PMF scheme outperforms the DP-PMF scheme on the three datasets. Further, due to the larger scale of ratings, the latter two datasets would prevent more noise interference and are much helpful to obtain stable results.

\section{Conclusions}

In this paper, we focus on users' individual privacy preferences for all items and build a personalized differential privacy recommendation scheme, PDP-PMF. The scheme is based on the probabilistic matrix factorization model, uses a modified sample mechanism with bounded DP and meets users' privacy requirements specified at the item level. Naturally, the PDP-PMF scheme would be very suitable for our real life.
To the best of our knowledge, we are the first to apply personalized differential privacy to the task of probabilistic matrix factorization.

Moreover, we propose a traditional differential privacy recommendation scheme DP-PMF that provides all users with a uniform level of privacy protection.
We confirm in theory that both schemes satisfy (personalized) differential privacy as desired.
Through our experiments, the PDP-PMF scheme improves greatly the recommendation quality compared with (traditional) DP-PMF.



\section*{Acknowledgements}
S. Zhang is partially supported by the National Natural Science Foundation of China (No. 11301002)
and Natural Science Foundation for the Higher Education Institutions of Anhui Province of China (KJ2018A0017).
Z. Chen is partially supported by the National Natural Science Foundation of China (No. 61572031). H. Zhong is partially supported by the National Natural Science Foundation of China (No. 61572001).



\begin{thebibliography}{99}

\setlength{\parsep }{-0.5ex}
\setlength{\itemsep}{-0.5ex}

\frenchspacing

\newcommand\BAMS{\emph{Bull. Amer. Math. Soc.\ }}
\newcommand\BIT{\emph{BIT\ }}
\newcommand\Com{\emph{Computing\ }}
\newcommand\CA{\emph{Constr. Approx.\ }}
\newcommand\FCM{\emph{Found. Comput. Math.\ }}
\newcommand\JAT{\emph{J. Approx. Th.\ }}
\newcommand\JC{\emph{J. Complexity\ }}
\newcommand\JMA{\emph{SIAM J. Math. Anal.\ }}
\newcommand\JMAA{\emph{J. Math. Anal. Appl.\ }}
\newcommand\JMM{\emph{J. Math. Mech.\ }}
\newcommand\MC{\emph{Math. Comp.\ }}
\newcommand\NM{\emph{Numer. Math.\ }}
\newcommand\RMJ{\emph{Rocky Mt. J. Math.\ }}
\newcommand\SJNA{\emph{SIAM J. Numer. Anal.\ }}
\newcommand\SR{\emph{SIAM Rev.\ }}
\newcommand\TAMS{\emph{Trans. Amer. Math. Soc.\ }}
\newcommand\TOMS{\emph{ACM Trans. Math. Software\ }}
\newcommand\USSR{\emph{USSR Comput. Maths. Math. Phys.\ }}

\addcontentsline{toc}{chapter}{Bibliography}


\bibitem{AG05}
Alessandro Acquisti and Jens Grossklags.
Privacy and rationality in individual decision making.
\emph{IEEE Security }\&\emph{ Privacy.} {\bf 3}, 26--33, 2005.

\bibitem{ABK16}
Friedman Arik, Shlomo Berkovsky and Mohamed Ali Kaafar.
A differential privacy framework for
matrix factorization recommender systems.
\emph{User Modeling and User-Adapted Interaction.}
{\bf 26}(5), 425--458, 2016.

\bibitem{AS11}
Frederik Armknecht and Thorsten Strufe.
An efficient distributed privacy-preserving recommendation system.
\emph{Ad Hoc NETWORKING Workshop. IEEE.}
65--70, 2011.

\bibitem{BEK07}
Shlomo Berkvosky, Yaniv Eytani, Tsvi Kuflik and Francesco Ricci.
Enhancing privacy and preserving accuracy of a distributed collaborative filtering.
\emph{Conference on Recommender Systems.}
9--16, 2007.


\bibitem{BGS05}
Berendt Bettina, Oliver G\"unther and Sarah Spiekermann.
Privacy in e-commerce: stated preferences vs. actual behavior.
\emph{Communications of the ACM.} {\bf 48}, 101--106, 2005.

\bibitem{Can02}
John Canny.
Collaborative filtering with privacy.
\emph{Security and Privacy, 2002. Proceedings. 2002 IEEE Symposium on. IEEE.}
{\bf 18}(1), 45--57, 2002.

\bibitem{CM09}
Kamalika Chaudhuri and Claire Monteleoni.
Privacy-preserving logistic regression.
\emph{International Conference on Neural Information Processing Systems.} 289--296, 2009.



\bibitem{Dwo08}
Cynthia Dwork.
Differential privacy: A survey of results.
\emph{International Conference on Theory and Applications of Models
 of Computation. Springer, Berlin, Heidelberg.} 1--19, 2008.

\bibitem{DMN06}
Cynthia Dwork, Frank McSherry, Kobbi Nissim and Adam Smith.
Calibrating noise to sensitivity in private data analysis.
\emph{Theory of Cryptography Conference. Springer, Berlin, Heidelberg.}
265--284, 2006.

\bibitem{DR14}
Cynthia Dwork and Aaron Roth.
The algorithmic foundations of differential privacy.
\emph{Foundations and Trends in Theoretical Computer Science.}
{\bf 9}(3--4), 211--407, 2014.

\bibitem{FM09}
McSherry Frank and Ilya Mironov.
Differentially private recommender systems:
Building privacy into the netflix prize contenders.
\emph{ACM SIGKDD International Conference on Knowledge Discovery and Data Mining ACM.}
627--636, 2009.

\bibitem{HR12}
Moritz Hardt, Aaron Roth. Beating randomized response on
incoherent matrices.\emph{Proceedings of the
Forty-fourth Annual ACM Symposium on Theory of
Computing}, 1255--1268, 2012.

\bibitem{HXZ15}
Jingyu Hua, Chang Xia and Sheng Zhong.
Differentially private matrix factorization.
\emph{IJCAI.} 1763--1770, 2015.

\bibitem{JYC15}
Zach Jorgensen, Ting Yu and Graham Cormode.
Conservative or liberal? personalized
differential privacy.
\emph{Data Engineering (ICDE), 2015 IEEE 31st
International Conference on. IEEE.}
1023--1034, 2015.

\bibitem{KR12}
Joseph A. Konstan and John Riedl.
Recommender systems: from algorithms to user experience.
\emph{User Modeling and User-Adapted Interaction.}
{\bf 22}(1--2), 101--123, 2012.

\bibitem{Kor10}
Yehuda Koren.
Factor in the neighbors: Scalable and accurate collaborative filtering.
\emph{ACM Transactions on Knowledge Discovery from Data (TKDD).}
{\bf 4}(1), 1, 2010.

\bibitem{LLS16}
Ninghui Li, Min Lyu, Dong Su and Weining Yang.
Differential Privacy: From Theory to Practice.
\emph{Synthesis Lectures on Information Security Privacy }\&\emph{ Trust.}
{\bf 8}(4), 1--138, 2016.


\bibitem{LLW17}
Yongkai Li, Shubo Liu, Jun Wang and Mengjun Liu.
A Local-Clustering-Based Personalized Differential Privacy Framework
for User-Based Collaborative Filtering.
\emph{International Conference on Database
Systems for Advanced Applications. Springer, Cham.}
543--558, 2017.

\bibitem{MYL08}
Hao Ma, Haixuan Yang, Michael R. Lyu and Irwin King.
Sorec: social recommendation using probabilistic matrix factorization.
\emph{ACM Conference on Information and Knowledge Management.}
931--940, 2008.

\bibitem{NIW13}
Valeria Nikolaenko, Stratis Ioannidis, Udi Weinsberg, et al.
Privacy-preserving matrix factorization. \emph{Proceedings of the
2013 ACM SIGSAC Conference on Computer \& Communications
Security}, 801--812, 2013.

\bibitem{PB07}
Rupa Parameswaran and Douglas M Blough.
Privacy Preserving Collaborative Filtering Using Data Obfuscation.
\emph{IEEE International Conference on Granular Computing. IEEE Computer Society.}
380, 2007.

\bibitem{POV17}
Kairouz Peter, Sewoong Oh and Pramod Viswanath.
The composition theorem for differential privacy.
\emph{IEEE Transactions on Information Theory}
{\bf 63}(6), 4037--4049, 2017.


\bibitem{SM08}
Ruslan R. Salakhutdinov and Andriy Mnih.
Probabilistic Matrix Factorization.
\emph{International Conference on Neural Information Processing Systems.}
 1257--1264, 2008.

\bibitem{WVR06}
Jun Wang, Arjen P. de Vries and Marcel J.T. Reinders.
Unifying user-based and item-based collaborative filtering approaches by similarity fusion.
\emph{Proc. of, ACM SIGIR Conference on Information Retrieval.} 501--508, 2006.

\bibitem{XWG11}
Xiaokui Xiao, Guozhang Wang and Johannes Gehrke.
Differential privacy via wavelet transforms.
\emph{IEEE Transactions on Knowledge and Data Engineering.}
{\bf 23}(8), 1200--1214, 2011.

\bibitem{YZX17}
Mengmeng Yang, Tianqing Zhu, Yang Xiang and Wanlei Zhou.
Personalized Privacy Preserving Collaborative Filtering.
\emph{International Conference on Green, Pervasive,
and Cloud Computing. Springer, Cham.} 371--385, 2017.







\bibitem{ZQZ17}
Zijian Zhang, Zhan Qin, Liehuang Zhu, Jian Weng and Kui Ren.
Cost-friendly differential privacy for smart meters:
Exploiting the dual roles of the noise.
\emph{IEEE Transactions on Smart Grid.}
{\bf 8}(2), 619--626, 2017.


\bibitem{ZRZ14}
Tianqing Zhu, Yongli Ren, Wanlei Zhou, Jia Rong and Ping Xiong.
An effective privacy preserving algorithm for
neighborhood-based collaborative filtering.
\emph{Future Generation Computer Systems.}
{\bf 36}, 142--155, 2014.

\end{thebibliography}
\end{document}